\documentclass[11pt,a4paper,reqno]{amsart}
\usepackage{amsthm,amsmath,amsfonts,amssymb,amsxtra,appendix,bookmark,dsfont,bm,mathrsfs,amstext,amsopn,mathrsfs,mathtools,comment,cite,color}
\usepackage[latin1]{inputenc}
\usepackage{dsfont}
\usepackage{hyperref}
\usepackage{comment,cite}
\usepackage{esint}
\usepackage{pst-all}
\usepackage{pstricks}
\usepackage{graphicx}
\usepackage{macros_art}
\usepackage{macros}
\newcommand{\no}[1]{ \left| \! \left| #1 \right| \! \right| }

\newcommand{\abs}[1]{\left|#1\right|}                                
\usepackage{ifthen}\newboolean{versionlongue}\setboolean{versionlongue}{false} 

\title[Unique continuation and the HK theorem II]{Unique continuation for \\ many-body Schr\"odinger operators \\ and the Hohenberg-Kohn theorem.
\\ II. The Pauli Hamiltonian}

\author[L. Garrigue]{Louis Garrigue}
\address{CEREMADE, Universit\'e Paris-Dauphine, PSL Research University, F-75016 Paris, France} 
\email{garrigue@ceremade.dauphine.fr}

\begin{document}
\begin{abstract} 
	We prove the strong unique continuation property for many-body Pauli operators with external potentials, interaction potentials and magnetic fields in $L^p\loc(\R^d)$, and with magnetic potentials in ${L^{q}\loc(\R^d)}$, where ${p > \max(2d/3,2)}$ and ${q > 2d}$. For this purpose, we prove a singular Carleman estimate involving fractional Laplacian operators.
\end{abstract}
\date{\today}

\maketitle
Density Funtional Theory (DFT) is the most successful approach to model matter at atomic and molecular scales. It is extensively employed to probe microscopic quantum mechanical systems, in very diverse situations. The one-body density of matter is the main object of interest in this framework. Indeed, a statement by Hohenberg and Kohn~\cite{HohKoh64}, lying at the heart of the theory, proves that, at equilibrium, the density contains all the information of the system. Later, Lieb~\cite{Lieb83b} showed that the rigorous proof of the Hohenberg-Kohn theorem relies on a strong unique continuation property (UCP).

Unique continuation is an important and versatile tool in analysis. In particular, it is used to prove uniqueness of Cauchy problems, see~\cite{Tataru04} for a review of some results. Unique continuation mainly relies on Carleman inequalities, first developed by Carleman~\cite{Carleman39}, later improved by H\"ormander~\cite{Hormander83} and Koch and Tataru~\cite{KocTat01}. It implies that, under general assumptions, a function verifying a second order partial differential equation and vanishing \apo{strongly} at one point vanishes everywhere. A famous result of this kind is due to Jerison and Kenig~\cite{JerKen85}, who dealt with eigenfunctions of Schr\"odinger's operator $-\Delta + V(x)$ where $V \in L^{n/2}\loc(\R^{n})$.

Nevertheless, most of the existing results fail to apply to situations that are relevant in many-body quantum physics, because their assumptions on potentials, which are generally $L^p$ conditions, depend on the number of particles $N$. For instance if we want to use the result of Jerison and Kenig, we need the electric potential to belong to $L^{dN/2}(\R^d)$, which is very restrictive when $N$ is large. The only two adapted works, having $N$-independent assumptions on the potentials, are the ones of Georgescu~\cite{Georgescu80} and Schechter-Simon~\cite{SchSim80}. But they hold only in a weak version, where it is assumed that the function vanishes in an open set. We also mention~\cite{Zhou12,Lammert18,Zhou19} on this subject, and finally~\cite{KinSha10} which goal is reached in this work.

This paper is a continuation of a previous article~\cite{Garrigue18}, where we showed the strong UCP for the many-body Schr\"odinger operator having external and interaction potentials. In this document as well, we replace $L^p$ conditions on potentials by relative boundedness with respect to the Laplacian, which is a classical assumption used in the analysis of Schr\"odinger operators. This enables us to extend our previous result~\cite{Garrigue18} to the important case of magnetic fields. Our proof relies on a Carleman inequality involving fractional Laplacians, which we prove using well-known techniques developed by H{\"o}rmander in~\cite{HormanderII83}
, further used by Koch and Tataru in~\cite{KocTat01}, and by R\"uland in~\cite{Ruland18}. This inequality pairs very naturally with Sobolev multipliers assumptions on the external potentials, which are independent of the number of particles. One of the difficulties with strong UCP results is that they need to use Carleman inequalities with singular weights. They are more delicate to show than for regular weights, because G\aa rding's inequality cannot be applied. We refer to~\cite{RouLeb12} for more details on Carleman estimates with regular weights. 

There are many works concerning unique continuation for Schr\"odinger operators with magnetic fields in the case of one-particle systems~\cite{Wolff90,Wolff93,Regbaoui97b,Regbaoui97,KocTat01,Davey14,ArrZub15}, using Carleman estimates. Another way of proving strong UCP results relies on techniques developed by Garofalo and Lin~\cite{GarLin86,GarLin87} which do not employ Carleman estimates but Almgren's monotonicity formula~\cite{Almgren79}. This was used by Kurata in~\cite{Kurata97}
to show strong UCP results for one-particle magnetic Schr\"odinger operators. Recently, Laestadius, Benedicks and Penz~\cite{LaeBenPen17} proved the first strong UCP result for many-body magnetic Schr\"odinger operators, using the work of Kurata. However, they need extra assumptions on $\pa{2V + x \cdot V}_-$ and $\rot A$, and a result with only $L^p$ hypothesis on potentials was lacking. 

The dimension of space being $d$, we can deal with external and interaction potentials as well as magnetic fields in $L^p\loc(\R^d)$ and magnetic potentials in $L^q\loc(\R^d)$, where
\begin{equation*}
\left\{
\begin{array}{ll}
  p > \max \pa{\frac{2d}{3},2}, \\
  q > 2d.
\end{array}
\right.
\end{equation*}
Our assumptions are independent of the number of particles $N$ and can treat the singular potentials involved in physics like the Coulomb one. Following Simon in \cite[Section C.9]{Simon82} and in light of~\cite{JerKen85,Wolff93,KocTat01}, we conjecture that the same results hold for $p = d/2$ if $d \ge 3$, $p > 1$ if $d= 2$ and $p=1$ if $d=1$, and $q = 2p$ for any dimension $d$. We tried to adapt the approach of~\cite{KocTat01} to the $N$-body setting, but did not manage to do so. We hope our work will stimulate further results in this direction. 

We also prove the strong UCP for the Pauli operator, which can be seen as an operator-valued matrix and thus belongs to the category of UCP results for systems of equations. Our result implies the Hohenberg-Kohn theorem in presence of a fixed magnetic field.

In order to take into account photons in a DFT context, Ruggenthaler and coworkers~\cite{RugFliPel14,Ruggenthaler15} considered the Pauli-Fierz operator together with a corresponding model where light and electrons are quantized, stating an adapted Hohenberg-Kohn theorem and calling the resulting theory QEDFT. Maxwell-Schr\"odinger theory is a variation of this hybrid model, in which photons are treated semi-classically through an internal self-generated magnetic potential $a$. Tellgren studied this model in~\cite{Tellgren18} within DFT and baptized the resulting framework Maxwell DFT. In a model describing external magnetic fields but not internal ones~\cite{VigRas87,VigRas88}, a generalization of the Hohenberg-Kohn theorem does not hold, counterexamples were provided in~\cite{CapVig02}. In DFT, an important problem has been to find a model bringing back this property~\cite{Diener91,PanSah10,VigUllCap13,PanSah14,TaoPanSah11,TelKvaSag12,LaeBen14,Laestadius14,PanSah15,TelLaeHel18}. The models containing internal magnetic potentials do so, as explained in~\cite{Ruggenthaler15,Tellgren18} and in this work, and our strong UCP result enables us to rigorously prove the Hohenberg-Kohn theorem in the Maxwell-Schr\"odinger model. Thus in this setting the one-body density $\ro$ and internal current $j + \rot m + \ro a$ of the ground state contain all the information of the system, that is the knowledge of the external classical electromagnetic field.

\subsection*{Acknowledgement}
I warmly thank Mathieu Lewin, my PhD advisor, for having supervized me during this work. This project has received funding from the European Research Council (ERC) under the European Union's Horizon 2020 research and innovation programme (grant agreement MDFT No 725528). This article was published under DOI 10.4171/DM/765, Documenta Mathematica, vol. 25, p. 869-898 (2020), and the journal version can be found at the address
\begin{center}
	\url{https://ems.press/journals/dm/articles/8965696}
\end{center}

\section{Main results}

Because it is of independent interest, we start by explaining the Carleman estimate which is the main tool of our approach.

\subsection{Carleman estimates for singular weights} \tx{ }

We denote by $B_R$ the ball of radius $R$ centered at the origin in $\R^n$, for $n \ge 1$. The first step in our study consists in a Carleman inequality obtained by standard techniques.

\begin{theorem}[Carleman inequality]\label{carl}
Let $0 < \alpha \le 1/2$, and let us define $\phi(x) \df -\ln \ab{x} + (-\ln \ab{x})^{-\alpha}$ for $\ab{x} \le 1/2$. In dimension $n$, there exist constants $c_{n}$ and $\tau_n \ge 1$ such that for any $\tau \ge \tau_n$ and any ${u \in C^{\ii}_{\tx{c}}(B_{1/2} \backslash \acs{0},\C)}$, we have
\begin{multline}\label{prem}
		 \tau^3 \int_{B_{1/2}} \f{\ab{e^{(\tau+2)\phi} u}^2}{\bigpa{\ln\ab{x}^{-1}}^{2+\alpha}} + \tau \int_{B_{1/2}} \f{\ab{e^{(\tau+1) \phi} \na u}^2}{\bigpa{\ln\ab{x}^{-1}}^{2+\alpha}}+ \tau \int_{B_{1/2}} \f{\ab{\na \pa{e^{(\tau+1) \phi} u}}^2}{\bigpa{\ln\ab{x}^{-1}}^{2+\alpha}} \\
		 +\tau^{-1} \int_{B_{1/2}} \f{ \ab{\Delta \bigpa{ e^{\tau\phi} u}}^2}{ \big( \ln\ab{x}^{-1} \big)^{2+\alpha}} \le \f{c_n}{\alpha} \int_{B_{1/2}} \ab{e^{\tau \phi} \Delta u}^2.
\end{multline}
\end{theorem}

Those are variants of known Carleman inequalities. The proof, given in Section \ref{sectioncarl}, follows from a rather standard reasoning. With $\phi$ a smooth pseudo-convex function, the classical Carleman estimate for \textit{regular} weights is
\begin{align}\label{classical}
\tau^{3} \nor{ e^{(\tau +2) \phi} u}{L^2}^2 + \tau \nor{e^{(\tau +1) \phi} \na u}{L^2}^2 + \tau^{-1} \nor{\Delta\pa{e^{\tau \phi} u}}{L^2}^2 \le c_{n}  \nor{e^{\tau \phi} \Delta u}{L^2}^2,
\end{align}
 for $\tau$ large enough, see~\cite{Tataru04,RouLeb12} for more detail.
 In~\cite{Regbaoui97b}, Regbaoui showed the estimate
\begin{align}\label{regb}
\tau^2 \nor{ \ab{x}^{-(\tau+2)} u}{L^2}^2 + \nor{ \ab{x}^{-(\tau+1)}   \na u}{L^2}^2 \le c_n \nor{\ab{x}^{-\tau} \Delta u}{L^2}^2,
\end{align}
where $\phi = \ln \ab{x}^{-1}$. This holds for $\tau \in \N + \ud$ which is a set preventing some quantity to intersect the spectrum of the Laplace-Beltrami operator on the sphere. The estimate \eqref{regb} is not good enough for us due to the slower increase of the coefficients in $\tau$.
In~\cite{Tataru04}, Tataru also presents a Carleman estimate with singular weights,
\begin{align}\label{class}
\tau^{3} \nor{ e^{(\tau +1) \phi} u}{L^2}^2 + \tau \nor{e^{\tau \phi} \na u}{L^2}^2 \le c_{n}  \nor{e^{\tau \phi} \Delta u}{L^2}^2,
\end{align}
 with 
\begin{align}\label{poids}
e^{\phi(x)} = \bigpa{\ab{x} + \lambda \ab{x}^2}^{-1},  
\end{align}
where $\lambda$ has to be negative for ${\phi}$ to be striclty convex, and where $u$ needs to be supported near the origin. Here the behavior in $\tau$ is optimal but the estimate on $\Delta u$ was not considered in~\cite{Tataru04}.\footnote{The published version of~\cite{Garrigue18} relies on this Carleman inequality. After publication of~\cite{Garrigue18}, we realized that we could not locate in the literature the same estimate on $\Delta\pa{e^{\tau \phi} u}$ as in \eqref{classical} for the weight \eqref{poids}, contrarily to what was stated in \cite[Theorem 1.1]{Garrigue18}. This article solves the problem and the needed \cite[Ineq. (8)]{Garrigue18} follows from Corollary \ref{fraccarl}.}
 Another Carleman estimate with singular weights was proved in~\cite{Ruland18}. It is similar to \eqref{class} and the weight function is
\begin{align*}
\phi(x) = - \ln \ab{x} + \f{1}{10} \pa{ (\ln \ab{x}) \arctan \ln \ab{x} - \ud \ln \pa{1+ (\ln \ab{x})^2}},
\end{align*}
for which $\phi(x) \sim -\pa{1+ \pi/20} \ln \ab{x}$ when $\ab{x} \ra 0^+$.

In our application to many-body Schr\"odinger operators, we needed a Carleman inequality having the best possible powers of $\tau$ outside the integrals, with a weight such that $\phi(x) \sim -\ln \ab{x}$ when $\ab{x} \ra 0^+$, for $e^{\phi}$ to be close enough to $\ab{\cdot}^{-1}$, and with the same powers of $e^{\phi}$ as in the classical estimate \eqref{classical}. Our inequality \eqref{prem} fulfills those requirements. The function $\phi$ in Theorem \ref{carl} respects 
\begin{align*}
\f{1}{\ab{x}} \le e^{\phi(x)} \le \f{e}{\ab{x}}.
\end{align*}
We obtain the same powers of $\tau$ as the classical estimate, and the singularity of the weight is the same as in the regular case, up to some logarithms. Defining $\vp(x) \df (-\ln \ab{x})^{-\alpha}$, the inequality \eqref{prem} can be rewritten as
\begin{multline*}
	\tau^3 \int_{B_{1/2}} \vp^{2(2+\alpha)} \ab{\f{e^{(\tau+2)\vp} u}{\ab{x}^{\tau+2}}}^2+ \tau \int_{B_{1/2}} \vp^{2(2+\alpha)} \ab{\f{e^{(\tau+1) \vp} \na u}{ \ab{x}^{\tau+1} } }^2  \\
		+ \tau \int_{B_{1/2}} \vp^{2(2+\alpha)} \ab{\na \pa{ \f{e^{(\tau+1) \vp} u}{ \ab{x}^{\tau+1} } }}^2  +\tau^{-1} \int_{B_{1/2}} \vp^{2(2+\alpha)} \ab{\Delta \pa{  \f{ e^{\tau\vp} u}{\ab{x}^{\tau}} }}^2  \\
		\le \f{c_n}{\alpha} \int_{B_{1/2}} \ab{ \f{ e^{\tau \vp} \Delta u}{ \ab{x}^{\tau} }}^2.
\end{multline*}

We transform now the inequality \eqref{prem} in a form tailored to be used in a very natural way for many-body operators. 

\begin{corollary}[Fractional Carleman inequality]\label{fraccarl}
	In dimension $n$, for any $\delta \in ]0,1]$ there exist constants $\kappa_n$ and $\tau_0 \ge 1$ such that for any $s \in \seg{0,1}$, ${s' \in \seg{0,\ud}}$, any $\tau \ge \tau_0$ and any $u \in C^{\ii}_{\tx{c}}(B_{1} \backslash \acs{0},\C)$, we have
 \begin{multline}\label{mainineq}
 \tau^{3-4s} \nor{ (-\Delta)^{(1-\delta)s} \pa{e^{\tau \phi} u}}{L^2(\R^n)}^2 + \tau^{1-4s'} \sum_{i=1}^n \nor{ \pa{-\Delta}^{(1-\delta)s'} \pa{e^{\tau \phi} \partial_i u}}{L^2(\R^n)}^2  \\
 \le \f{\kappa_{n}}{\delta^{5/2}}  \nor{e^{\tau \phi} \Delta u}{L^2(B_1)}^2.
 \end{multline}
 \end{corollary}

The function $\phi$ is the same as in Theorem \ref{carl}. The constant $\kappa_{n}$ depends only on the dimension $n$. The proofs of Theorem~\ref{carl} and Corollary~\ref{fraccarl} are provided later in Section~\ref{sectioncarl}.

\subsection{Unique continuation properties} \tx{ }

We state a strong unique continuation property (UCP) result for Schr\"odinger operators involving gradients, in which potentials are Sobolev multipliers. This type of assumption pairs very naturally with the Carleman inequality involving fractional Laplacians \eqref{mainineq}, as we can see from the proof. At the same time, those assumptions will allow us to prove a corresponding result for the many-body Pauli operator.
\begin{theorem}[Strong UCP for systems with gradients]\label{sucp}
Let $\delta >0$ (small), let $\wt{V} \df \pa{V_{\alpha,\beta}}_{1 \le \alpha,\beta \le m}$ be a $m \times m$ matrix of potentials in $L^2_{\rm{loc}}(\er{n},\C)$ and let $\wt{A} \df \pa{A_{\alpha}}_{1 \le \alpha \le m}$ be a list of vector potentials in $L^2_{\rm{loc}}(\er{n},\R^n)$, such that for every ${R >0}$, there exists $c_R \geq 0$ such that 
	\begin{equation}\label{hyps}
\left\{
\begin{array}{rl}
 \indic_{B_R} \abs{V_{\alpha,\beta}}^2 & \leq \ep_{n,m,\delta} (-\Delta)^{\frac{3}{2} - \delta}  + c_R, \\
\indic_{B_R} \abs{A_{\alpha}}^2 & \leq \ep_{n,m,\delta} (-\Delta)^{\frac{1}{2} - \delta} + c_R, \\
	\ab{\ps{ u, -i \indic_{B_R}  A_{\alpha}\cdot \na u}}  & \leq \ep_{n,m} \ps{u, \pa{(-\Delta) + c_R}u}, \hs\hs\forall u \in \cC^{\ii}\ind{c}(\R^n),
\end{array}
\right.
\end{equation}
	in $\R^n$ in the sense of quadratic forms, where $\ep_{n,m,\delta}$ and $\ep_{n,m}$ are small constants depending only on their indices.
	Let $\p \in H_{\rm{loc}}^2(\er{n},\C^{m})$ be a weak solution of the $m \times m$ system
	\begin{align}\label{system}
	 \pa{- \indic_{m \times m} \Delta_{\R^n} + i \wt{A} \cdot \na_{\R^n} + \wt{V}} \p = 0,
 \end{align}
	where $\wt{A} \cdot  \na_{\R^n}$ is the $m \times m$ operator-valued matrix $\diag \pa{A_{\alpha} \cdot \na_{\R^n}}_{1 \le \alpha \le m}$.
	If $\p$ vanishes on a set of positive measure or if it vanishes to infinite order at a point, then ${\p=0}$.
\end{theorem}
In all this document, when we write $L \le J$ for two symmetric operators $L$ and $J$, we mean it in the sense of forms. We recall that $\p$ vanishes to infinite order at $x_0 \in \er{n}$ when for all $k \geq 1$, there is a $c_k$ such that
 \begin{equation}\label{vanish}
	 \int_{\abs{x-x_0} < \ep} \abs{\p}^2 \d x < c_k \ep^k,
 \end{equation}
 for every $\ep <1$.

Let $A$ be a magnetic potential and $B$ a magnetic field. Physically in dimension 3, $A$ and $B$ are linked by $B = \rot A$, but we will consider arbitrary dimensions and artificially remove the link between $A$ and $B$. We consider the $N$-particle Pauli Hamiltonian
 \begin{multline}\label{op}
	  H^N (v,A,B)   \df \\
	 \sum_{\ell=1}^N \pa{ \pa{-i\na_{\ell} + A(x_{\ell})}^2 + \sigma_{\ell} \cdot B(x_{\ell}) + v(x_{\ell})} + \sum_{1 \le t \sle \ell \le N} w(x_t-x_{\ell}),
 \end{multline}
 where $\sigma_{\ell}$ are generalizations of Pauli matrices. They are $d$ square matrices of size $2^{\floor{\pa{d-1}/{2}}} \times 2^{\floor{\pa{d-1}/{2}}}$ used to form the $(d+1)$-dimensional chiral representation of the Clifford algebra, which structures Lorentz-invariant spinor fields 
 \cite[Appendix E]{WitSmi12}. 
As an operator-valued matrix, the only non-diagonal member is the Stern-Gerlach term $\sum_{\ell=1}^N \sigma_{\ell} \cdot B(x_{\ell})$, responsible for the Zeeman effect. We refer to \cite[Chapter XII and Complement $A_{\tx{XII}}$]{CohDiuLal86} for a discussion on this Hamiltonian. The previous theorem implies the strong UCP for this operator, which is our main result.

\begin{corollary}[Strong UCP for the many-body Pauli operator]\label{sucppauli}
	Let $\delta >0$ and assume that the potentials satisfy 
\begin{align} \label{borne}
\pa{\abs{v}^2  + \abs{w}^2 + \abs{B}^2 + \ab{\div A}^2 }\indic_{B_R} \leq \ep_{d,N} (-\Delta)^{\f{3}{2}-\delta} + c_{R} \qquad \tx{ in $\er{d}$}, \\
\abs{A}^2 \indic_{B_R}  \leq \ep_{d,N} (-\Delta)^{\f{1}{2}-\delta} + c_{R} \qquad \tx{ in $\er{d}$},
\end{align}
	for all $R>0$, where $\ep_{d,N}$ is a small constant depending only on $d$ and $N$. For instance $A \in L^q\loc (\R^d,\R^d)$ and ${\ab{B}, \div A, v,w \in L^p_{\rm{loc}} (\er{d},\R)}$ where
\begin{equation}\label{pq}
\left\{
\begin{array}{ll}
  p > \max \pa{\frac{2d}{3},2}, \\
  q > 2d.
\end{array}
\right.
\end{equation}
	Let $\p \in H_{\rm{loc}}^2(\er{dN})$ be a solution to $H^N(v,A,B) \p = 0$. If $\p$ vanishes on a set of positive measure or if it vanishes to infinite order at a point, then ${\p=0}$.
\end{corollary}
The proof of this corollary is the same as the one of \cite[Corollary 1.2]{Garrigue18}. In particular, this result can be applied to the magnetic Schr\"odinger operator $H^N(v,A,0)$. In the Appendix we recall how assumptions of this Corollary \ref{sucppauli} imply assumption of Theorem \ref{sucp} on the gradient term.

\subsection{Hohenberg-Kohn theorems in presence of magnetic fields} \tx{ }

We give here two applications of our strong UCP result in Density Functional Theory. The first one is the classical Hohenberg-Kohn theorem in presence of a fixed magnetic field.

\subsubsection{Fixed magnetic fields} 
In presence of one spin internal degree of freedom, the one-particle density and the paramagnetic current of a wave function $\p$ are respectively defined by
 \begin{align*}
	 \ro_{\p}(x) & \df \sum_{(s_k)_{1 \le k \le N} \in \acs{\upa,\doa}^N} \sum_{i=1}^N \int_{\R^{d(N-1)}} \ab{\p^{s_k}}^2
	 \d x_1 \cdots \d x_{i-1} \d x_{i+1} \cdots \d x_N, \\
	 j_{\p}(x) & \df \im \hspace{-0.8cm} \sum_{(s_k)_{1 \le i \le N} \in \acs{\upa,\doa}^N} \sum_{i=1}^N \int_{\R^{d(N-1)}} \ov{\p^{s_k}} \na_i \p^{s_k}  \d x_1 \cdots \d x_{i-1} \d x_{i+1} \cdots \d x_N.
 \end{align*}

\begin{theorem}[Hohenberg-Kohn with a fixed magnetic field]\label{hkthm}
	Let $A \in (L^{q}+L^{\ii})(\er{d},\R^d)$, $B \in (L^{p}+L^{\ii})(\er{d},\R^d)$ and $w, v_1, v_2 \in (L^{p}+L^{\ii})(\er{d},\reals)$, with $p$ and $q$ as in \eqref{pq}. If there are two normalized eigenfunctions $\p_1$ and $\p_2$ of $H^N(v_1,A,B)$ and $H^N(v_2,A,B)$, corresponding to the first eigenvalues, and such that $\ro_{\p_1} = \ro_{\p_2}$, then there exists a constant $c$ such that $v_1 = v_2 + c$. 
\end{theorem}
The proof is the same as in the standard case where $A=B=0$, the only difference is that we need to use Corollary \ref{sucppauli} to justify that the nodal set of the ground states have zero measures. We refer to the same arguments as in~\cite{HohKoh64,Lieb83b,PinBokLud07,Garrigue18}.

\subsubsection{Ill-posedness of the Hohenberg-Kohn theorem for Spin-Current DFT} 
We recall the definition of Pauli matrices in dimension 3,
 \begin{align*}
	  \sigma^x = \begin{pmatrix} 0 & 1 \\ 1 & 0 \end{pmatrix}, \qquad \sigma^y = \begin{pmatrix} 0 & -i \\ i & 0 \end{pmatrix}, \qquad \sigma^z = \begin{pmatrix} 1 & 0 \\ 0 & -1 \end{pmatrix},
 \end{align*}
 they act on one-particle two-component wavefunctions $\phi = \begin{pmatrix} \phi^{\upa} &  \phi^{\doa} \end{pmatrix}^{T}$, where $\phi^{\upa}, \phi^{\doa} \in L^2(\er{d},\C)$ and $\int \abs{\phi}^2 = 1$. We denote by $L^2_{\tx{a}}(\R^{dN})$ the space of antisymmetric functions of $N$ variables in $\R^d$. The state of a system is described by wavefunctions $\p \in L^2_{\tx{a}}(\R^{dN},\C^{2^N})$. We introduce the one-body densities
 \begin{align*}
\ro^{\alpha \beta}_{\p}(x) \df \sum_{s \in \acs{\upa,\doa}^{N-1}} \sum_{i=1}^N \int_{\er{d(N-1)}} \p^{\alpha,s}(x,Y) \ov{\p^{\beta,s}}(x,Y) \d Y,
 \end{align*}
 where $\alpha, \beta \in \acs{\upa, \doa}$. We remark that $\ro_{\p}^{\upa\doa} = \ov{\ro_{\p}^{\doa\upa}} \eqdef \xi_{\p}$.
We define the density $\ro_{\p} \df \ro_{\p}^{\upa \upa} + \ro_{\p}^{\doa\doa}$ and the locally gauge invariant magnetization
 \begin{align*}
 m_{\p} \df \begin{pmatrix} \ro_{\p}^{\upa \doa} + \ro_{\p}^{\doa\upa} \\ -i\pa{\ro_{\p}^{\upa \doa} - \ro_{\p}^{\doa\upa}} \\ \ro_{\p}^{\upa \upa} - \ro_{\p}^{\doa\doa} \end{pmatrix} = \mat{ 2\re \xi_{\p} \\ 2 \im \xi_{\p} \\ \ro_{\p}^{\upa \upa} - \ro_{\p}^{\doa\doa}}.
 \end{align*}
 The energy of a quantum wavefunction is coupled to the magnetic field only through the density $\ro_{\p}$ and through the magnetization current $j_{\p}+\rot m_{\p}$. Indeed, using either bosonic or fermionic statistics,
 \begin{multline*}
	 \ps{\p,\sum_{\ell=1}^N \bigpa{\sigma_{\ell} \cdot \bigpa{-i\na_{\ell} + A(x_{\ell})}}^2 \p}\\
	 = \int \ab{\na \p}^2 + \int A^2 \ro_{\p} +  \int A \cdot (2 j_{\p} +\rot m_{\p}).
 \end{multline*}
Previously, the fields $A$ and $B$ were independent. We now assume the physical relation $B = \rot A$, take the Coulomb gauge $\div A = 0$ and consider the physical Hamiltonian $H^N (v,A)  \df  H^N(v,A,\rot A)$.

A natural question is whether the model with Pauli operator and varying magnetic fields has a corresponding Hohenberg-Kohn theorem, i.e. we ask whether $(\ro_{\p_1}, 2j_{\p_1} + \rot m_{\p_1}) =(\ro_{\p_2},2 j_{\p_2} + \rot m_{\p_2})$ (or even $(\ro_{\p_1},j_{\p_1},m_{\p_1})$ $=(\ro_{\p_2},j_{\p_2},m_{\p_2})$) implies $A_1 = A_2$ and $v_1 = v_2 +c$. This turns out to be wrong, due to counterexamples found by Capelle and Vignale in~\cite{CapVig02}.

Many authors studied this ill-posedness issue~\cite{Diener91,PanSah10,CapVig02,VigUllCap13,PanSah14,TaoPanSah11,TelKvaSag12,LaeBen14,Laestadius14,PanSah15,TelLaeHel18}, also from the point of view of Spin DFT, in which current effects are neglected~\cite{KohSha65,BarHed72,RajCal73,EscPic01,CapVig01,KohSavUll05,PanSah15,ReiBorTel17}, and from the point of view of Current DFT, in which spin effects are neglected~\cite{Diener91,PanSah10,TelKvaSag12,VigUllCap13,LaeBen14,PanSah14,LaeBen15}. In particular, see Laestadius and Benedicks in \cite[Theorem 2]{LaeBen14} for a counterexample in Current DFT.

Nevertheless, one could try to find a similar result using the physical total current, that is the one which can be measured in experiments, $j_{\tx{t}} \df j + \rot m + \ro A$ and respects $\div j_{\tx{t}} = 0$. As explained in~\cite{PanSah10,TelKvaSag12}, for one particle and for Current DFT where $j_{\tx{t}} \df j + \ro A$, the relation $\rot (j_{\tx{t}}/ \ro) = \rot A$ shows that the knowledge of $j_{\tx{t}}$ and $\ro$ gives the knowledge of $A$ and $v$ by the Hohenberg-Kohn theorem. The case of $N \ge 2$ particles is still open.

\subsubsection{Hohenberg-Kohn for the Maxwell-Schr\"odinger model}
We keep the dimension $d=3$. In order to get a model taking into account current effects but having a Hohenberg-Kohn theorem, and as a second application of our strong UCP result, we follow Tellgren~\cite{Tellgren18} and investigate the Maxwell-Schr\"odinger theory. This is a hybrid model of quantum mechanics where electrons are treated quantum mechanically and light is treated classically, and provides an approximation of non-relativistic QED~\cite{Heisenberg34,Erdos05}. It was studied through a DFT approach in~\cite{Tellgren18} and the resulting framework was called Maxwell DFT.
We define 
 \begin{align*}
\cA(\R^d,\R^d) \df \acs{ A \in H^1(\R^d,\R^d) \bigst \div A = 0 \tx{ weakly in } H^1(\R^d) },
 \end{align*}
the set of divergence-free magnetic potentials, i.e. potentials in the Coulomb gauge. A state of matter and light is given by a pair $(\p,a) \in L^2_{\tx{a}}(\R^{dN},\C^{2^N}) \times \cA$, where $\p$ describes electrons and where $a$ is an internal magnetic potential describing the photon cloud around the electrons.
 
We denote by $H^N_0 \df H^N(0,0)$ the kinetic and interaction parts of the Schr\"odinger operator. 
 The energy functional takes into account the energy of $\p$ coupled to the total magnetic field, and the kinetic energy of the internal magnetic field. We denote by $\alpha$ the fine structure constant and define ${\ep \df (8\pi \alpha^2)^{-1}}$. The Maxwell-Schr\"odinger energy functional is
 \begin{align*}
	 \cE_{v,A}(\p,a) & \df \ps{\p, H^N(v,a+A) \p} + \ep \int \ab{\rot a}^2 \\
	 & = \ps{\p,H_0^N \p} + \int  \pa{v+ \ab{a+A}^2} \ro_{\p}\\
	 &\bighs  + \int (a+A) \cdot (2j_{\p}+ \rot m_{\p}) +\ep \int \ab{\rot a}^2,
 \end{align*}
 for bosons or fermions. We denote by 
 \begin{align*}
	 E \df \myinf{\p \in (H^1 \cap L^2_{\tx{a}}) (\R^{dN},\C) \\ \int \ab{\p}^2 = 1 \\ a \in (L^q\loc \cap \cA) (\R^d,\C^d)} \cE_{v,A}(\p,a),
 \end{align*}
the ground state energy. This functional was studied in~\cite{LosYau86,Erdos05,LieLosSol95} when $A=0$, where the authors found that in the case of a Coulomb potential generated by only one atom having a large number of protons, this minimum was $-\ii$. In \cite[Theorem 1]{LieLosSol95}, they also prove that for Coulomb potentials induced by molecules, if the total number of protons in the molecule is lower than $1050$, independently of the positions of the nucleus, then the functional is bounded below. This justifies the applicability of the next theorem to physical systems. When one removes the Zeeman term $\sum_{\ell=1}^N \sigma_{\ell} \cdot B(x_{\ell})$ and considers the corresponding functional, then this issue disappears and the functional is always bounded from below for $v,w \in (L^{d/2}+L^{\ii})(\R^d)$ and $A \in (L^d+L^{\ii})(\R^d)$, by the diamagnetic inequality.

 The corresponding Euler-Lagrange equations are the Schr\"odinger equation together with a Maxwell equation. Using $\rot^* = \rot$ and $\rot \rot = \na \div - \Delta$ we can show that if it exists, the ground state $(\p,a)$ respects
 \begin{align*}
	    \left\{
      \begin{aligned}
	      & \sum_{\ell=1}^N \pa{-\Delta_{\ell} -2i  (a+A) \cdot \na_{\ell} +v+ \ab{a+A}^2 + \sigma_{\ell} \cdot (\rot (a +A))_{\ell}} \p = E \p, \\
	       &  j_{\p}+ \rot m_{\p} + \ro_{\p} (a +A)- \ep \Delta a =0,
      \end{aligned}
    \right.
 \end{align*}
 where we did not write all the $x_j$ arguments in the first equation for simplicity of notation.
The internal current of a state $(\p,a)$ is defined by
 \begin{align*}
 j_{(\p,a)} \df j_{\p} + \rot m_{\p} + \ro_{\p} a.
 \end{align*} 
We remark that if we did not fix the gauge $\div A = 0$, $j_{(\p,a)}$ would be locally gauge invariant. We make an other preliminary remark on the density of solutions of Schr\"odinger's equation.
\begin{remark}\label{rq}
Let $\p$ be a solution of $H^N(v,A,B)\p = 0$, under the assumptions of Corollary \ref{sucppauli}. Then its density vanishes almost nowhere, $$\ab{\acs{x \in \R^d \st \ro_{\p}(x)=0}}=0.$$ 
\end{remark}
Indeed, if $\ro_{\p}$ vanishes on a set $S \subset \R^d$ of positive measure, then $$N \int_{S \times \R^{d(N-1)}} \ab{\p}^2 = \int_S \ro_{\p} = 0$$ so $\p$ vanishes on $S \times \R^{d(N-1)}$ which has infinite volume. But by the strong UCP theorem for Pauli operators, Corollary \ref{sucppauli}, $\p$ does not vanish on sets of positive measure. This holds in any spin number, in particular this holds when there is no spin. 
We are now ready to prove the Hohenberg-Kohn theorem for this model. 

\begin{theorem}[Hohenberg-Kohn for Maxwell DFT]\label{mdft}
Let $p > 2$ and $q >6$ and let $w,v_1,v_2\in (L^p+L^{\ii})(\R^3,\R)$, $A_1, A_2 \in \bigpa{L^q\loc\cap \cA}(\R^3,\R^3)$ be potentials such that $\cE_{v_1,A_1}$ and $\cE_{v_2,A_2}$ are bounded from below and admit lowest energy states $(\p_1,a_1)$ and $(\p_2,a_2)$. If $\ro_{\p_1} = \ro_{\p_2}$ and $j_{(\p_1,a_1)} = j_{(\p_2,a_2)}$, then $A_1 = A_2$ and there is a constant $c$ such that $v_1 = v_2 + c$.
\end{theorem}
This result shows that in the Maxwell-Schr\"odinger framework, the knowledge of the ground state density $\ro$ and internal current $j+ \rot m + \ro a$ gives the knowledge of $v$ and $A$. Said differently, at equilibrium, $\ro$ and $j+ \rot m+ \ro a$ contain the information of $v$ and $A$. This is a rigorous justification of Tellgren's Hohenberg-Kohn theorem~\cite{Tellgren18}.
\begin{proof}
	Let us denote by $\ro \df \ro_{\p_1} = \ro_{\p_2}$ the common densities, by $j' \df j_{(\p_1,a_1)} = j_{(\p_2,a_2)}$ the common internal currents, and by $E_i \df \cE_{v_i,A_i}(\p_i,a_i)$ for $i \in \acs{1,2}$ the ground state energies. By the standard proof of the Hohenberg-Kohn theorem~\cite{HohKoh64,Garrigue18}, we can prove that $\cE_{v_1,A_1}(\p_2,a_2) = E_1$. So $(\p_2,A_2)$ respects the Euler-Lagrange equations for $\cE_{v_1,a_1}$, that is
 \begin{align*}
	    \left\{
      \begin{aligned}
	      & \sum_{\ell=1}^N \pa{-\Delta_{\ell} -2i  (a_2+A_1) \cdot \na_{\ell} +v_1+ \ab{a_2+A_1}^2 + \sigma_{\ell} \cdot (\rot (a_2+A_1))_{\ell} } \p_2 \\
	      & \bighs \bighs\bighs\bighs\bighs\bighs\bighs\bighs\bighs\bighs= E_1 \p_2, \\
	      &  j' + \ro A_1- \ep \Delta a_2  = 0.
      \end{aligned}
    \right.
 \end{align*}
	We take the difference of those equations with the Euler-Lagrange equations verified by $(\p_2,a_2)$ for $\cE_{v_2,A_2}$ and get
	\begin{align}\label{el}
	    \left\{
      \begin{aligned}
	      & E_2-E_1 + \sum_{\ell=1}^N \bigpa{ -2i (A_1-A_2) \cdot \na_{\ell} + \sigma_{\ell} \cdot (\rot (A_1-A_2))_{\ell} }\p_2 \\
	      & \bighs\bighs\bighs +\sum_{\ell=1}^N\bpa{v_1-v_2 + \ab{a_2+A_1}^2 - \ab{a_2+A_2}^2 }(x_{\ell}) \p_2 = 0, \\
	      &  \ro (A_1-A_2)=0.
      \end{aligned}
    \right.
 \end{align}
	By the unique continuation result of Corollary \ref{sucppauli}, and by Remark \ref{rq}, $\ro$ does not vanish on sets of positive measure, therefore the second equation in \eqref{el} yields $A_1=A_2$. Using it in the first equation of \eqref{el}, we get 
 \begin{align*}
	 \biggpa{E_2-E_1 + \sum_{\ell=1}^N \pa{v_1-v_2}(x_{\ell})} \p_2 = 0
 \end{align*}
	so, by the same argument as in~\cite{Garrigue18}, we conclude that $v_1 = v_2 + (E_1-E_2)/N$.
 \end{proof}

One could still want to search for a Hohenberg-Kohn theorem in the standard Schr\"odinger model but involving the knowledge of $j_{\tx{t}}$ instead of the knowledge of $j$. This is an open problem. Our result easily extends to the same model but without spin effects, that is when we take for the one-body kinetic operator $(-i \na + A)^2$ instead of $\pa{\sigma \cdot \pa{-i \na + A}}^2$. Then the internal current is $j + \ro a$ and the above results hold.

\section{Proofs of Carleman inequalities}\label{sectioncarl}

\subsection{Proof of Theorem \ref{carl}}

	 We use standard arguments which can for instance be read in~\cite{KocTat01,Ruland18}. We denote by $r \df \ab{x}$ the radial coordinate. In dimension $n$, the Laplace operator in spherical coordinates is $\Delta = \partial_{rr} + \f{n-1}{r}\partial_r + \inv{r^2} \Delta_S$, where $\Delta_S$ is the Laplace-Beltrami operator on the $(n-1)$-dimensional sphere. Using log-spherical coordinates $t \df \ln r$, we have
 \begin{align*}
	 \abs{x}^2 \Delta = \partial_{tt} + (n-2)\partial_t + \Delta_S.
 \end{align*}
We take the function $\phi(x) = -\ln \ab{x} + (-\ln \ab{x})^{-\alpha}$ as in the statement of the Theorem, and define $\vp(t) \df \phi(e^t)$. More explicitly,
 \begin{align*}
	 \vp(t) & = -t + \inv{(-t)^{\alpha}}, \bighs  \vp'(t) = -1 + \f{\alpha}{(-t)^{\alpha+1}}, \bighs  \vp''(t) = \f{\alpha(\alpha+1)}{(-t)^{\alpha+2}}, \\
	 & \vp'''(t)  = \f{\alpha(\alpha+1)(\alpha+2)}{(-t)^{\alpha+3}}, \bighs \vp''''(t) =  \f{\alpha(\alpha+1)(\alpha+2)(\alpha+3)}{(-t)^{\alpha+4}},
 \end{align*}
	 so $-1 < \vp' < -1/8$ and $\vp'',\vp''',\vp'''' > 0$ on $]-\ii,-\ln 2]$. 
Conjugating the previous operator $\abs{x}^2 \Delta$ with $e^{\tau \phi}$ yields
 \begin{align*}
	 P \df e^{\tau \phi} \abs{x}^2 \Delta e^{-\tau \phi} = \partial_{tt} + \pa{-2\tau \vp' +n-2}\partial_t+ \tau^2 \vp'^2 -\tau (n-2) \vp'  + \Delta_S,
 \end{align*}
 and decomposing the result in symmetric and antisymmetric parts, we have $P = S + A$, where
 \begin{align*}
	 S & \df \partial_{tt} + \tau^2 \varphi'^2  - \tau (n-2)\vp' + \tau \vp''+ \Delta_S, \\
	 A & \df \pa{-2\tau \vp' + n-2} \partial_t - \tau \vp''.
 \end{align*}
	 We implicitly take the $L^2(\R^n)$ norm. 
	 We want to manipulate $\no{Pv}^2 = \no{Sv}^2 + \no{Av}^2 + \ps{v,\seg{S,A}v}$ for a function $v \in C^{\ii}_{\tx{c}}\pa{]-\ii,-\ln 2] \times \bbS^{n-1},\C}$. We compute 
 \begin{align*}
	 \seg{S,A} = & -4 \tau \vp'' \partial_{tt} - 2 \tau \vp''' \partial_t \\
	 &  - \pa{-2\tau \vp' + n-2}\pa{2\tau^2 \vp' \vp'' - \tau (n-2) \vp'' + \tau \vp'''}- \tau \vp''''.
 \end{align*}
 We have $2 \re{\ps{Sv,Av}} = \ps{v,\seg{S,A}v}$ so this term is real and integrating by parts yields
 \begin{multline*}
	 \ps{v,\seg{S,A}v} = 4 \tau^3 \int \vp'^2 \vp'' \ab{v}^2 + 2 \tau^2 \int \vp' \pa{\vp'''-n \vp''} \ab{v}^2  \\
	+ 4 \tau \int \vp'' \abs{\partial_t v}^2 + \tau (n-2)^2 \int \vp'' \ab{v}^2  - 2 \tau \int \vp'''' \ab{v}^2 - \tau (n-2) \int \vp''' \ab{v}^2.
 \end{multline*}
Thus for $\tau$ large enough,
 \begin{align}\label{partialt}
4 \tau^3 \int \vp'^2 \vp'' \ab{v}^2 + 4 \tau \int \vp'' \ab{\partial_{t} v}^2 \le  \ps{v,\seg{S,A}v} \le\no{Pv}^2.
 \end{align}
With $\ab{\ps{\vp'' v, Sv}} \le \no{\vp'' v}\no{Sv} \le \tau^{-\f{3}{2}}\no{Pv}^2/2$, we compute the radial part of the gradient
 \begin{align*}
	 \int \vp''  \ab{\na_S v}^2 & = \ps{ \vp''v, (-\Delta_S)v} \\
	 & =  \ps{ \vp''v, \pa{-S+ \partial_{tt} + \tau^2 \varphi'^2  - \tau (n-2)\vp' + \tau \vp'' }v} \\
	 & = \tau^2 \int \vp'^2 \vp'' \ab{v}^2 -\tau(n-2) \int \vp' \vp'' \ab{v}^2+ \tau \int \vp'' \ab{v}^2 \\
	 & \bighs  +  \ud \int \vp'''' \ab{v}^2 - \ps{\vp'' v, Sv}-\int \vp'' \ab{\partial_t v}^2  \\
	  & \le \tau^2 \int \vp'^2 \vp'' \ab{v}^2- \ps{\vp'' v, Sv} \le \inv{2 \tau} \no{Pv}^2,
 \end{align*}
	 for $\tau$ large enough.
 Now, using the inequality \eqref{partialt} again, we find
	 \begin{align}\label{thi}
	 \tau^3 \int \vp'^2 \vp'' \ab{v}^2 + \tau \int \vp'' \pa{ \ab{\partial_t v}^2 + \ab{\na_S v}^2} \le \no{Pv}^2.
 \end{align}
	 Working back in cartesian coordinates, we have 
 \begin{align*}
\ab{\partial_t v}^2 + \ab{\na_S v}^2 = \ab{x}^2 \ab{\na v}^2.
 \end{align*}
	 
	 We can now apply the previous well-known techniques to form the inequality \eqref{prem}, which is fitted with our application. Defining $u \df e^{\tau \phi} v$ and using 
	 \begin{align}\label{estno}
		  1 \le \ab{x}e^{\phi} \le e,
	 \end{align}
	 inequality \eqref{thi} implies
 \begin{align*}
	 \tau^3 \int \f{\ab{e^{(\tau+2) \phi} u}^2}{(-\ln \ab{x})^{2+\alpha}}  + \tau \int \f{\ab{x}^2 \ab{\na \pa{e^{(\tau+2) \phi} u}}^2}{(-\ln \ab{x})^{2+\alpha}} \le \f{2^6 e^4}{\alpha} \int \ab{ e^{\tau \phi} \Delta u}^2.
 \end{align*}
	 Using also 
	 \begin{align}\label{estgr}
 1 \le \ab{x} \ab{\na \phi} \le e^{(\ln 2)^{-1/2}} \le 4,
	 \end{align}
in $B_{1/2}$, we have 
	 \begin{align*}
		 \ab{e^{(\tau+1)\phi} \na u}^2 & = \ab{e^{-\phi} \pa{e^{(\tau+2) \phi} \na u}}^2 \\
		 & = \ab{e^{-\phi} \na \pa{ e^{(\tau+2) \phi} u} - (\tau+2) e^{-\phi}(\na \phi) e^{(\tau+2) \phi} u}^2 \\
		 & \le 2 e^{-2\phi} \ab{\na \pa{ e^{(\tau+2) \phi} u}}^2 + 2(\tau+2)^2  e^{-2\phi}\ab{\na \phi}^2 \ab{e^{(\tau+2) \phi} u}^2 \\
		 & \le 2 \ab{x}^2 \ab{\na \pa{e^{(\tau+2)\phi} u}}^2 + 2^5 (\tau+2)^2 \ab{e^{(\tau+2)\phi} u}^2,
	 \end{align*}
and similarly
 \begin{align*}
	 \ab{\na \pa{ e^{(\tau+1) \phi} u}}^2 & = \ab{ (\tau+1) \pa{ \na \phi} e^{(\tau+1) \phi} u + e^{(\tau+1) \phi} \na u}^2 \\
	 & \le 2^5(\tau+1)^2 \ab{e^{(\tau+2) \phi} u}^2 +2 \ab{e^{(\tau+1) \phi} \na u}^2.
 \end{align*}
 Eventually, we obtain
 \begin{multline*}
	 \tau^3 \int \f{\ab{e^{(\tau+2) \phi} u}^2}{(-\ln \ab{x})^{2+\alpha}}  + \tau \int \f{\ab{ e^{(\tau+1)\phi} \na u}^2}{(-\ln \ab{x})^{2+\alpha}} + \tau \int \f{\ab{ \na \pa{e^{(\tau+1)\phi} u}}^2}{(-\ln \ab{x})^{2+\alpha}} \\
	 \le \f{2^{14} e^4}{\alpha} \int \ab{ e^{\tau \phi} \Delta u}^2.
 \end{multline*}
 These are the first terms in \eqref{prem}. We now turn to the estimates on the second derivative.
 Since $\ab{x}^2 \Delta \phi = (n-2)  \vp'(\ln \ab{x}) + \vp''(\ln \ab{x})$ we have 
 \begin{align}\label{estlap}
 \ab{x}^2 \ab{\Delta \phi} \le \ab{n-2} + \f{3}{4(\ln 2)^{3/2}} \le n+4.
 \end{align}
 Since 
 \begin{align}\label{grads}
	 \na e^{\tau \phi} = \tau e^{\tau \phi} \na \phi, \bhs  \Delta e^{\tau \phi} = \tau e^{\tau \phi} \pa{ \Delta \phi + \tau \ab{\na \phi}^2},
 \end{align}
then we find
 \begin{align*}
	 \tau^{-1} \int \f{\ab{ \Delta \pa{e^{\tau\phi} u}}^2}{(-\ln \ab{x})^{2+\alpha}}  \le \f{2^{25} e^4(n+4)^2}{\alpha} \int \ab{ e^{\tau \phi} \Delta u}^2.
 \end{align*}
The constant $c_n$ in \eqref{prem} can be taken to be $2^{25} e^4(n+4)^2$, for instance.
\qed

\subsection{Proof of Corollary \ref{fraccarl}}

	 We denote by $c$ a constant which only depends on the dimension $n$. We fix $\alpha = 1/2$. For any $a \in ]0,1[$, we have
	 \begin{align*}
		 e^{-2a\phi} \le \f{c}{a^{\f{5}{2}} (-\ln\ab{x})^{\f{5}{2}}},
	 \end{align*}
 on $B_{1/2}$.
	 So the inequality \eqref{prem} taken from {Theorem \ref{carl}} implies
	 \begin{align}\label{mh}
	 & \tau^3 \no{e^{(\tau+2-a) \phi} u}^2+  \tau \no{e^{(\tau+1-a) \phi} \na u}^2  \\
	 & \bighs \le c \tau^3 a^{-\f{5}{2}} \no{\f{e^{(\tau+2) \phi} u}{(-\ln \ab{x})^{\f{5}{4}}}}^2+ c \tau a^{-\f{5}{2}}  \no{\f{e^{(\tau+1) \phi} \na u}{(-\ln \ab{x})^{\f{5}{4}}}}^2  \nonumber\\
	 & \bighs \leq c a^{-\f{5}{2}} \no{e^{\tau \phi} \Delta u}^2. \nonumber
 \end{align}
	 We now compute
 \begin{align*}
	 & \ab{x}^{a}  \ab{\Delta \pa{e^{\tau \phi} u}} \\
	 & \bhs = \ab{x}^{a} \ab{ u\Delta e^{\tau \phi} + 2 \na u \cdot \na e^{\tau \phi} + e^{\tau \phi} \Delta u} \\
	 & \bhs= \ab{x}^{a} \ab{  \tau e^{\tau \phi} u \Delta \phi  + \tau^2  \ab{\na \phi}^2 e^{\tau \phi} u +2\tau \na \phi \cdot e^{\tau \phi} \na u + e^{\tau \phi} \Delta u} \\
	 & \bhs \le \tau^2 e^{\tau \phi} \ab{u} \pa{ \ab{x}^a \ab{\Delta \phi} + \ab{x}^a \ab{\na \phi}^2} + 2 \tau e^{\tau \phi} \ab{\na u} \ab{x}^a \ab{\na \phi} \\
	 & \bhs \bhs \bhs + e^{\tau \phi} \ab{\Delta u} \ab{x}^a.
 \end{align*}
	 Since we work in $B_{1/2}$, we have $\ab{x}^a \le 1$. Using \eqref{estgr} yields $\ab{x}^a \ab{\na \phi}^2 \le 16 \ab{x}^{a-2}$ and $\ab{x}^a \ab{\na \phi} \le 4 \ab{x}^{a-1}$, and \eqref{estlap} yields $\ab{x}^a \ab{\Delta \phi} \le (n+4) \ab{x}^{a-2}$, so we get
	 \begin{align}\label{gh}
	 & \ab{x}^{a} \ab{\Delta \pa{e^{\tau \phi} u}} \\
	 & \bhs \le  (n+20) \tau^2 \ab{x}^{a-2} e^{\tau  \phi} \ab{u} + 8 \tau \ab{x}^{a-1} e^{\tau\phi} \ab{\na u} +  e^{\tau \phi} \ab{\Delta u} \nonumber \\
	  & \bhs\le (n+20) e^{a-2} \tau^2 e^{(\tau +2-a) \phi} \ab{u} + 8e^{a-1} \tau e^{(\tau+1-a)\phi} \ab{\na u} +  e^{\tau \phi} \ab{\Delta u} \nonumber\\
	  & \bhs\le c \pa{ \tau^2 e^{(\tau +2-a) \phi} \ab{u} + \tau e^{(\tau+1-a)\phi} \ab{\na u} +  e^{\tau \phi} \ab{\Delta u}},\nonumber
 \end{align}
	 where in the second inequality we applied \eqref{estno}. We will also use the fractional Hardy inequality,
	 \begin{align}\label{hardy}
		 \nor{(-\Delta)^{-\delta} \ab{x}^{-2\delta}}{L^2(\R^n) \ra L^2(\R^n)} = 4^{-\delta} \pa{\f{\Gamma\pa{\f{n-2\delta}{4}} }{ \Gamma\pa{\f{n+2\delta}{4}}}}^2 \le 1,
	 \end{align}
	  which holds for any $\delta \in [0,n/2[$. Its sharp constant was found in~\cite{Herbst77,Beckner95a,Yafaev99}. Choosing $a = \delta/2 \in \seg{0,n/4}$, we are ready the deduce that
 \begin{align*}
	 & \no{ (-\Delta)^{1-\f{a}{2}} \pa{e^{\tau \phi} u}} \\
	 & \bighs =\no{(-\Delta)^{-\f{a}{2}} \ab{x}^{-a} \ab{x}^a  (-\Delta) \pa{e^{\tau \phi} u}} \\
	 & \bighs\le \no{\ab{x}^{ a} (-\Delta)\pa{e^{\tau \phi} u}} \\
	 &\bighs \le  c \pa{ \tau^2 \no{e^{(\tau+2-a) \phi} u } + \tau \no{e^{(\tau+1-a) \phi} \na u} + \no{e^{\tau \phi} \Delta u}} \\
	 &\bighs \le  c a^{-\f{5}{4}} \tau^{\ud} \no{e^{\tau \phi} \Delta u},
 \end{align*}
	 where, in the inequalities, we respectively used \eqref{hardy}, \eqref{gh} and \eqref{mh}.
	 Applying H\"older's inequality together with $$\no{ e^{\tau \phi} u} \le ca^{-\f{5}{4}} \tau^{-\f{3}{2}}\no{e^{\tau \phi} \Delta  u},$$ as implied by \eqref{prem}, yields the first part of our claim \eqref{mainineq}. We remark that this is also true for $a \in [n/4,1[$.

	 Now we show the second part of the inequality.
	  We begin by expanding
 \begin{align*}
	 & \ab{x}^a \ab{\partial_i \pa{e^{\tau\phi} \partial_j u}} \\
	 & \hs \hs\hs= \ab{x}^a\Bigr\vert \partial_{ij} \pa{e^{\tau \phi} u} - \tau  \pa{\partial_{ij} \phi }e^{\tau \phi} u -\tau^2 \pa{\partial_j \phi} \pa{\partial_i \phi} e^{\tau \phi} u -\tau \pa{\partial_j \phi} e^{\tau \phi} \partial_i  u \Bigr\vert \\
	 &\hs \hs\hs \le  c \tau^2 e^{(\tau+2-a) \phi} \ab{u} + c \tau e^{(\tau+1-a) \phi}  \ab{\partial_i  u} + \ab{\partial_{ij} \pa{e^{\tau \phi} u}}.
\end{align*}
	 Therefore by \eqref{hardy},
 \begin{align*}
	 & \no{(-\Delta)^{\ud-\f{a}{2}} \pa{e^{\tau \phi} \partial_j u}} \\
	 & \bighs = \no{(-\Delta)^{-\f{a}{2}} \na \pa{e^{\tau \phi} \partial_j u}} \\
	 & \bighs \le \no{\ab{x}^{a} \na \pa{e^{\tau \phi} \partial_j u}} \\
	 &\bighs  \le c \pa{\tau^2 \no{e^{(\tau+2-a) \phi} u} + \tau \no{ e^{(\tau+1-a) \phi} \na u} + \no{\partial_{ij} \pa{e^{(\tau-a) \phi} u}} } \\
	 &\bighs  \le c \pa{\tau^2 \no{e^{(\tau+2-a) \phi} u} + \tau \no{ e^{(\tau+1-a) \phi} \na u} + \no{\Delta \pa{e^{(\tau-a) \phi} u}} } \\
	 &\bighs  \le c a^{-\f{5}{4}} \tau^{\ud} \no{e^{\tau \phi} \Delta u},
 \end{align*}
where we used $2 \ab{k_i k_j} \le k_i^2 + k_j^2$.
	 Applying H\"older's inequality together with $$\no{e^{\tau \phi} \partial_j u} \le c a^{-\f{5}{4}} \tau^{-\f{3}{2}} \no{e^{\tau \phi} \Delta u},$$ we obtain the second part of the sought-after inequality \eqref{mainineq}.
	 \qed

\section{Proof of the strong unique continuation property}
We present here the proof of Theorem \ref{sucp}, which follows rather closely that in~\cite{Garrigue18}.

\subsection*{Step 1. Vanishing on a set of positive measure implies vanishing to infinite order at one point.}

To prove that $\p$ vanishes to infinite order at a point, we extend a property showed by Figueiredo and Gossez in~\cite{FigGos92}, to magnetic fields.
\begin{proposition}[Figueiredo-Gossez with magnetic term]\tx{ }

Let $V \in L^1_{\rm{loc}} (\er{n},\C)$ and $A \in L^1_{\rm{loc}} (\R^n,\R^n)$ such that for every $R >0$, there exist $a,a'$ and $c >0$ such that ${a+a' \sle 1}$ and 
\begin{align*}
	-\indic_{B_R} \re V & \leq a (-\Delta) +c, \\
	 \ps{u,-i \indic_{B_R} A \cdot \na u} & \leq \ps{u,\pa{a' (-\Delta) +c} u}, \hs\hs \forall u \in \cC^{\ii}\ind{c}(\R^n).
\end{align*}
Let $\p \in H^1_{\rm{loc}}(\er{n})$ satisfying $-\Delta \p + i A \cdot \na \p + V \p=0$ weakly. If $\p$ vanishes on a set of positive measure, then $\p$ has a zero of infinite order.
\end{proposition}
\begin{proof}
	We take $\delta \in (0,1/2]$ and define a smooth real positive localisation function $\eta$ with support in $B_{2\delta}$, equal to $1$ in $B_{\delta}$, and such that $\abs{\na \eta} \leq c/\delta$ and $\abs{\Delta \eta} \leq c/\delta^2$. And integration by parts and the use of $2 \re \ov{\p} \na \p = \na \ab{\p}^2$ yields
	\begin{align*}
	 \re \int \eta^2 \ov{\p} \Delta \p & = - \int \ab{\eta \na \p}^2 - \int \na (\eta^2) \re \ov{\p} \na \p \\
	 & = - \int \ab{\eta \na \p}^2 + \ud \int \ab{\p}^2 \Delta (\eta^2).
 \end{align*}
Hence, multiplying Schr\"odinger's equation by $\eta^2 \ov{\p}$, taking the real parts, integrating by parts and rearranging the obtained equation yields
	\begin{align}\label{grad}
 \int \ab{\eta \na \p}^2 & = - \re  \int \eta^2 \ov{\p} i A \cdot  \na \p -\re \int V \ab{\eta \p}^2 + \ud \int \ab{\p}^2 \Delta \eta^2\nonumber  \\
		& = -  \int \eta\ov{\p} i A \cdot  \na \pa{ \eta \p}  - \int \pa{ \re V} \ab{\eta \p}^2 + \ud \int \ab{\p}^2 \Delta \eta^2\nonumber  \\
		& \le (a+a')\int \ab{\na (\eta \p)}^2 + 2c \int \ab{\eta \p}^2 + \f{1}{2} \int \ab{\p}^2 \Delta \eta^2\nonumber  \\
		& = (a+a')\hspace{-0.1cm}\int \ab{\eta \na \p}^2 + (a+a')\hspace{-0.1cm} \int \ab{\p \na \eta}^2 + \f{1-a-a'}{2} \hspace{-0.1cm}\int \ab{\p}^2 \Delta \eta^2\nonumber  \\
  & \bighs + 2c \int \ab{\eta \p}^2,
\end{align}
where we used the assumptions on the potentials. We move the first term of the right-hand-side to the left, which yields
\begin{align}\label{eses}
\int \ab{ \eta \na \p}^2 \le c \int \ab{\p}^2 \pa{ \eta^2 + \ab{\na \eta}^2 + \ab{\Delta \eta^2}} \le \f{c}{\delta^2} \int_{B_{2\delta}} \ab{\p}^2,
 \end{align}
where we used that $\supp \eta \subset B_{2\delta}$. Then we have
 \begin{align*}
	 \int_{B_\delta} \ab{\na \p}^2 = \int_{B_{\delta}} \ab{\eta \na \p}^2 \le \int \ab{\eta \na \p}^2 \le \f{c}{\delta^2} \int_{B_{2\delta}} \ab{\p}^2,
 \end{align*}
where $c$ is independent of $\delta$, and where we used that $\eta = 1$ in $B_{\delta}$. This estimate is the same statement as \cite[Lemma 1]{FigGos92}. The end of the proof is thus exactly the proof of \cite[Proposition 3]{FigGos92}. This consists in applying H\"older's and Sobolev's inequalities, so $\int_{B_{\delta}} \ab{\p}^2$ is controlled by $\int_{B_{2\delta}} \ab{\p}^2$ times a factor which is proved to be small by using Lebesgue's density theorem. Iterating this estimate yields \eqref{vanish}, that is the definition of $\p$ vanishes to infinite order at the origin.
\end{proof}
The last proof extends to Pauli operators, which Zeeman part can be put in a matrix potential.
\begin{proposition}[Figueiredo-Gossez for Pauli systems]\label{fig-gos}\tx{ }

	Let $\wt{V} \df \pa{V_{\alpha,\beta}}_{1 \le \alpha,\beta \le m}$ be a $m \times m$ matrix of potentials in $L^2_{\rm{loc}}(\er{n},\C)$ and let $\wt{A} \df \pa{A_{\alpha}}_{1 \le \alpha \le m}$ be a list of vector potentials in $L^2_{\rm{loc}}(\er{n},\R^n)$, such that for every ${R >0}$, there exists $c_R \geq 0$ such that 
\begin{align*}
-\indic_{B_R} \re V_{\alpha,\beta} & \leq \ep_{n,m} (-\Delta)  + c_R,\\
 \ps{u,-i\indic_{B_R} A_{\alpha} \cdot \na u} & \leq \ps{u,\pa{\ep_{n,m} (-\Delta) + c_R}u}, \hs\hs \forall u \in \cC^{\ii}\ind{c}(\R^n),
\end{align*}
where $\ep_{n,m}$ is a small constant depending only on the dimensions $n$ and $m$.
Let $\p \in H_{\rm{loc}}^2(\er{n},\C^{m})$ be a weak solution of the $m \times m$ system \eqref{system}, that is
\begin{align*}
 \pa{- \indic_{m \times m} \Delta_{\R^n} + i \wt{A} \cdot \na_{\R^n} + \wt{V}} \p = 0.
 \end{align*}
If $\p$ vanishes on a set of positive measure, then $\p$ has a zero of infinite order.
\end{proposition}
Without loss of generality, we can thus assume that $\p$ vanishes to infinite order at the origin. 

\subsection*{Step 2. $\na \p$ and $\Delta \p$ vanish to infinite order as well.}

As remarked in \cite[Section 2, Step 2]{Garrigue18}, if $\p \in L^2(\er{n})$, then vanishing to infinite order at the origin is equivalent to $\int_{B_1}\abs{x}^{-\tau} \abs{\p}^2 \d x$ being finite for every $\tau \geq 0$. With additional assumptions, we can show that $\na \p$ and $\Delta \p$ vanish to infinite order as well.

\begin{lemma}[Finiteness of weighted norms]\label{finite} \tx{ }

$i)$ If $\p \in H^{1+\ep}\loc(\R^n)$ with $\ep > 0$ and if $\p$ vanishes to infinite order at the origin, then $\na \p$ as well.

$ii)$ 
Let $V \in L^2_{\rm{loc}} (\er{n},\C)$ and $A \in L^2_{\rm{loc}} (\er{n},\R^n)$ be such that 
\begin{align*}
-\indic_{B_1} \re V & \leq a (-\Delta) +c, \\
	\ps{u,- i \indic_{B_1} A \cdot \na u} & \leq  \ps{u,\pa{a' (-\Delta) +c'}u}, \hs\hs \forall u \in \cC^{\ii}\ind{c}(\R^n),
\end{align*}
for some $a,a'$ such that $a+a'<1$ and $c,c' \ge0$.
Let $\p \in H^1_{\rm{loc}}(\er{n})$ satisfying $-\Delta \p + i A \cdot \na \p + V \p=0$ weakly. If $\p$ vanishes to infinite order at the origin, then $\na \p$ as well.

	$iii)$ 
	Let $V \in L^2_{\rm{loc}} (\er{n},\C)$ and $A \in L^2_{\rm{loc}} (\er{n},\C^n)$ be such that 
	\begin{align*}
		\abs{V}^2 \indic_{B_1}  \leq a (-\Delta)^2 +c, \bighs \bighs \abs{A}^2 \indic_{B_1}  \leq \ep (-\Delta) +c_{\ep},
	\end{align*}
	for some $a<1$, $c \ge 0$, for all $\ep> 0$ and some $c_{\ep} \ge 0$ depending on $\ep$.
	Let $\p \in H^2_{\rm{loc}}(\er{n})$ satisfying $-\Delta \p + i A \cdot \na \p + V \p=0$. If $\p$ vanishes to infinite order at the origin, then $\na \p$ and $\Delta \p$ as well.
\end{lemma}
This lemma also extends to Pauli operators.

\begin{proof}
	$i)$ Take $\delta \in ]0,1/4[$, and choose a smooth real positive localization function $\eta$ equal to $1$ in $B_{\delta} \subset \er{n}$, supported in $B_{2\delta}$, and such that $0 \leq \eta \leq 1$, $\ab{\na \eta} \le c / \delta$, and $\ab{\Delta \eta} \le c / \delta^2$. 
	For any $k \in \N$, $k \ge 0$, there exists $c_k \ge 0$ such that
	\begin{align*}
		\int_{B_{\delta}} \abs{\na \p}^2 & = \int_{B_{\delta}} \abs{\na \pa{ \eta \p}}^2 \leq \int \abs{\na \pa{\eta \p}}^2 \leq c \nor{(-\Delta)^{\f{1+\ep}{2}} \pa{\eta \p}}{L^2}^{\f{1}{1+\ep}} \nor{\eta \p}{L^2}^{\f{\ep}{1+\ep}} \\
		& \le \f{c \nor{\p}{H^{1+\ep}(B_{2\delta})}^{\f{1}{1+\ep}}}{ \delta^{2}} \pa{\int_{B_{2\delta}} \ab{\p}^2}^{\f{\ep}{1+\ep}} \le \f{c}{\delta^2} \pa{c_k \pa{2\delta}^k}^{\f{\ep}{1+\ep}} = c_k' \delta^{\f{k\ep}{1+\ep}-2},
 \end{align*} 
	where we applied Gagliardo-Nirenberg's inequality in the second inequality, we used $\p \in H^{1+\ep}\loc(\R^n)$ in the following one, and we used the definition of $\p$ vanishing to infinite order at the origin \eqref{vanish} in the last inequality. We notice that our estimate is the definition of $\na \p$ vanishing to infinite order at the origin.

	$ii)$ Let $\eta$ be the same function as in $i)$, and consider the inequality \eqref{grad} again, in which we used Schr\"odinger's equation of $\p$, we obtained \eqref{eses}, and we will use it once more. Since $\eta = 1$ in $B_{\delta}$, we have
 \begin{align*}
	 \int_{B_{\delta}} \abs{\na \p}^2 = \int_{B_{\delta}} \abs{\eta \na \p}^2\leq \int \abs{\eta \na \p}^2 \leq \f{c}{ \delta^{2}} \int_{B_{2\delta}} \abs{\p}^2 \le \f{c}{\delta^2} \pa{c_k (2\delta)^k} = c_k' \delta^k,
 \end{align*}
where we used the definition of $\p$ vanishing to infinite order, and this proves that $\na \p$ vanishes to infinite order as well.

	$iii)$ We take the same funtion $\eta$ as in $i)$, adding the constraint $\ab{\partial_{ij} \eta}<c/r^2$ for any $i,j \in \acs{1,\dots,n}$, and we take $\delta \in ]0,1/4[$. We know that for any $\xi,\theta \in \R$ and any $\alpha \in ]0,+\ii[$, we have
 \begin{align*}
	 \pa{\xi + \theta}^2 \le \pa{1+ \alpha}\xi^2 + \pa{1+ \alpha\iv} \theta^2.
 \end{align*}
	So using the assumption on $V$, we have
 \begin{align*}
	 & \int \abs{V \eta \p}^2 \\
	 & \bhs \leq a \int \ab{\Delta(\eta \p)}^2 + c \int \ab{\eta \p}^2 \\
	 &  \bhs = a \int \ab{\eta \Delta \p + 2 \na \eta \cdot \na \p + \p \Delta \eta}^2 + c \int \ab{\eta \p}^2 \\
	 &  \bhs \le a(1+\alpha) \int \ab{\eta \Delta \p}^2 + \pa{1+ \inv{\alpha}} \int \ab{ 2\na \eta \cdot \na \p + \p \Delta \eta}^2 + c \int \ab{\eta \p}^2 \\
	  & \bhs \le a(1+\alpha) \int \abs{\eta \Delta \p}^2  + 2\pa{1+ \inv{\alpha}} \int \abs{\p \Delta \eta}^2\\
	 & \bhs \bhs \bhs \bhs  + 4\pa{1+ \inv{\alpha}} \int \abs{\na \p \cdot \na \eta}^2 + c\int \abs{\eta \p}^2,
 \end{align*}
	for any $\alpha >0$.
	As for the gradient term, we have
 \begin{align*}
	 \int \abs{\eta A \cdot \na \p}^2 & \le \int \ab{A}^2 \ab{ \eta \na \p}^2 \\
	 & \le \ep \int \ab{\na  \ab{\eta \na \p}}^2 + c_{\ep} \int \ab{\eta \na \p}^2 \\
	 & = \ep \int \ab{\eta\na \ab{ \na \p} + \ab{\na \p} \na \eta}^2 + c_{\ep} \int \ab{\eta \na \p}^2 \\
	 & \le 2 \ep \int \eta^2 \ab{ \na \ab{ \na \p}}^2 + 2\ep \int \ab{\na \eta}^2 \ab{\na \p}^2 + c_{\ep} \int \ab{\eta \na \p}^2.
 \end{align*}
	We denote by $\na^2 \p = \pa{\partial_{ij} \p}_{1 \le i,j \le n}$ the Hessian of $\p$, its square being $\ab{ \na^2 \p}^2 = \sum_{1 \le i,j \le n} \ab{ \partial_{ij} \p}^2$.
	Now by convexity of the map $f \mapsto \ab{\na \sqrt{f}}^2$ and then the diamagnetic inequality, we have
 \begin{align*}
	 \ab{\na \ab{\na \p}}^2 = \ab{ \na \sqrt{\sum_{i=1}^n \ab{\partial_i \p}^2}}^2 \le \sum_{i=1}^n \ab{\na \ab{\partial_i \p}}^2 \le \sum_{i=1}^n \ab{\na \partial_i \p}^2 = \ab{ \na^2 \p}^2.
 \end{align*}
	Also, 
 \begin{align*}
	 & \int \ab{ \na^2 \pa{\eta \p}}^2 = \sum_{1 \le i,j \le n} \int \ab{k_i k_j \hat{\eta \p}}^2 \\
	 & \bighs \le \ud \sum_{1 \le i,j \le n} \int \pa{\ab{k_i}^2 + \ab{ k_j}^2} \ab{\hat{\eta \p}}^2  = n \int \ab{ \Delta \pa{ \eta \p}}^2,
 \end{align*}
	therefore, denoting by $\otimes$ the tensor product on $n \times n$ matrices, and making use of previous inequalities, we obtain
 \begin{align*}
	 \int \ab{\eta \na^2 \p}^2 & = \int \ab{\na^2 \pa{\eta \p} - \p \na^2 \eta - \na \eta \otimes \na \p - \na \p \otimes \na \eta}^2 \\
	 & \le 4 \int \ab{ \na^2 \pa{\eta \p}}^2 + 4 \int \ab{\p \na^2 \eta}^2 + 8 \int \ab{\na \eta \otimes \na \p}^2 \\
	 & \le 4n \int \ab{ \Delta \pa{\eta \p}}^2 + 4 \int \ab{\p \na^2 \eta}^2 + 4 n^2 \int \ab{\na \eta}^2 \ab{\na \p}^2 \\
	 & \le 4n \int \ab{ \eta \Delta \p}^2 +4n \int \ab{ \p \Delta \eta}^2+ 4 \int \ab{\p \na^2 \eta}^2 \\
	 & \bighs\bighs \bighs \bighs \bighs  + 4 n(n+2) \int \ab{\na \eta}^2 \ab{\na \p}^2.
 \end{align*}
	We use Schr\"odinger's equation pointwise and gather our previous inequalities. We get, for any $\alpha, \beta > 0$,
 \begin{align*}
	 \int &\abs{\eta \Delta \p}^2 = \int \ab{\eta V \p + i \eta A \cdot \na \p}^2 \\
	 & \le (1+ \beta) \int \abs{V \eta \p}^2 + \pa{1+ \inv{\beta}}\int \abs{\eta A \cdot \na \p}^2 \\
	  & \le  \pa{a(1+\beta)(1+\alpha) + 8\ep n \pa{1+ \inv{\beta}}} \int \ab{\eta \Delta \p}^2 \\
	 & + \pa{ 2\pa{1+ \inv{\alpha}}(1+\beta) + 8 \ep n \pa{1+ \inv{\beta}}} \int \ab{\p \Delta \eta}^2 \\
	 & + \pa{ 4 \pa{1+ \inv{\alpha}}(1+\beta) + 2 \ep \pa{4n(n+2) +1} \pa{1+ \inv{\beta}}} \int \ab{\na \eta}^2 \ab{\na \p}^2 \\
	 & + (1+\beta)c \int \ab{\eta \p}^2 + c_{\ep} \pa{1+ \inv{\beta}}\int \ab{\eta \na \p}^2 + 8 \ep \pa{1+ \inv{\beta}} \int \ab{\p \na^2 \eta}^2.
 \end{align*}
 We take $\alpha,\beta $ and $\ep$ such that $a(1+\beta)(1+\alpha) + 8\ep n \pa{1+ \beta\iv}  < 1$. This allows us to move the term $\int \abs{\eta \Delta \p}^2$ to the left and obtain
 \begin{align*}
	 \int_{B_{\delta}} \abs{\Delta \p}^2 \leq \int \abs{\eta \Delta \p}^2 \leq \f{c}{\delta^4} \int_{B_{2\delta}} \pa{\abs{\p}^2+\abs{\na \p}^2} \le \f{c}{\delta^4} \pa{c_k (2\delta)^k} = c_k' \delta^{k-4}.
 \end{align*}
	This proves that $\Delta \p$ vanishes to infinite order at the origin, by the definition \eqref{vanish}.
\end{proof}

\subsection*{Step 4. Proof that $\p =0$.}

We consider some number $\tau\geq 0$ (large), and we call $c$ any constant which does not depend on $\tau$. We take a smooth localisation function $\eta$, equal to $1$ in $B_{1/2} \subset \er{n}$, supported in $B_{1}$, and such that $0 \leq \eta \leq 1$. We take the same weight function $\phi$ as in Theorem~\ref{carl}. Thanks to Step 3, all the expressions we write are finite. 
We define $\kappa_{\delta,n} \df \kappa_n \delta^{-3}$. We start by controling the gradient term by using the assumption on $\tilde{A}$,
 \begin{align*}
	 \nor{e^{\tau \phi}  \wt{A} \cdot \na \pa{\eta \p}}{L^2(B_1)}^2 &  = \sum_{\alpha=1}^m \nor{e^{\tau \phi} \sum_{i=1}^n A^i_{\alpha} \partial_i \pa{ \eta \p_{\alpha}}}{L^2(B_1)}^2 \\
	 & \le n  \sum_{\substack{1 \le \alpha \le m \\ 1 \le i \le n}} \nor{e^{\tau \phi}  A^i_{\alpha} \partial_i \pa{ \eta \p_{\alpha}}}{L^2(B_1)}^2  \\
	 & \le n m \ep_{n,m,\delta}\sum_{\substack{1 \le \alpha \le m \\ 1 \le i \le n}}  \nor{ (-\Delta)^{\f{1}{4} - \delta} \pa{ e^{\tau \phi}  \partial_i \pa{ \eta \p_{\alpha}}}}{L^2(B_1)}^2 \\
	 & \hs\hs\hs\bhs  + n m c \sum_{\substack{1 \le \alpha \le m \\ 1 \le i \le n}} \nor{e^{\tau \phi} \partial_i \pa{ \eta \p_{\alpha}}}{L^2(B_1)}^2.
 \end{align*}
We now use the fractional Carleman inequality \eqref{fraccarl} with $s'=1/4$ and $s'=0$, this yields
 \begin{align*}
	 & \nor{e^{\tau \phi}  \wt{A} \cdot \na \pa{\eta \p}}{L^2(B_1)}^2 \\
	 & \bhs \le  \kappa_{n} n m^2 \pa{\ep_{n,m,\delta}\pa{\delta/4}^{-\f{5}{2}} + \tau^{-1} c } \sum_{\alpha=1}^m  \nor{  e^{\tau \phi}  \Delta \pa{ \eta \p_{\alpha}}}{L^2(B_1)}^2  \\
	 & \bhs =  \kappa_{n} n m^2 \pa{\ep_{n,m,\delta}\pa{\delta/4}^{-\f{5}{2}} + \tau^{-1} c } \nor{  e^{\tau \phi}  \Delta \pa{ \eta \p}}{L^2(B_1)}^2.
 \end{align*}
  Similarly, for the multiplication potential $\tilde{V}$, we begin by using the assumption \eqref{hyps} and we get
 \begin{align*}
	 & \nor{e^{\tau \phi} \eta \wt{V} \p}{L^2(B_1)}^2 \\
	 & \bhs = \sum_{\alpha=1}^m \nor{e^{\tau \phi} \eta \sum_{\beta=1}^m V_{\alpha\beta} \p_\beta}{L^2(B_1)}^2 \\
	  &  \bhs \le m \sum_{\substack{1 \le \alpha, \beta \le m}} \nor{e^{\tau \phi} \eta V_{\alpha\beta} \p_\beta}{L^2(B_1)}^2 \\
	  & \bhs \le  m \sum_{\beta=1}^m \pa{\ep_{n,m,\delta} \nor{ (-\Delta)^{\f{3}{4}-\delta}  \pa{e^{\tau \phi} \eta \p_\beta}}{L^2(B_1)}^2 + c \nor{ e^{\tau \phi} \eta \p_{\beta} }{L^2(B_1)}^2 }.
 \end{align*}
We proceed by using our fractional Carleman inequality of Corollary~\ref{fraccarl}, yielding
 \begin{align*}
	 \nor{e^{\tau \phi} \eta \wt{V} \p}{L^2(B_1)}^2  & \le  \kappa_n m \pa{ \ep_{n,m,\delta}(3\delta/4)^{-\f{5}{2}} + \tau^{-3} c } \sum_{\beta=1}^m  \nor{ e^{\tau \phi} \Delta \pa{\eta \p_{\beta}}}{L^2(B_1)}^2  \\
	 &  =  \kappa_n m \pa{ \ep_{n,m,\delta}(3\delta/4)^{-\f{5}{2}} + \tau^{-3} c }  \nor{ e^{\tau \phi} \Delta \pa{\eta \p}}{L^2(B_1)}^2.
 \end{align*}
 We can now estimate
 \begin{align*}
	  \nor{ e^{\tau \phi} \eta \Delta \p }{L^2(B_1)}^2 &  = \sum_{\alpha=1}^m \nor{e^{\tau \phi} \eta \Delta \p_{\alpha}}{L^2(B_1)}^2 \\
	 &  \le 2 \nor{e^{\tau \phi} \eta \wt{A} \cdot \na \p}{L^2(B_1)}^2 + 2 \nor{e^{\tau \phi} \eta \wt{V} \p}{L^2(B_1)}^2 \\
	 &  = 2 \nor{e^{\tau \phi} \wt{A} \cdot \pa{\na (\eta \p) - \p \na \eta}}{L^2(B_1)}^2 + 2 \nor{e^{\tau \phi} \eta \wt{V} \p}{L^2(B_1)}^2 \\
	 & \le 4 \nor{e^{\tau \phi} \wt{A} \cdot \na \pa{ \eta \p}}{L^2(B_1)}^2 + 4 \nor{e^{\tau \phi} \p \wt{A} \cdot \na \eta}{L^2(B_1)}^2 \\
	 & \bighs +2 \nor{e^{\tau \phi} \eta \wt{V} \p}{L^2(B_1)}^2 \\
	 &  \le 6\kappa_{n} n m^2 \pa{\ep_{n,m,\delta}\pa{\delta/4}^{-\f{5}{2}} + \tau^{-1} c }\nor{  e^{\tau \phi}  \Delta \pa{ \eta \p}}{L^2(B_1)}^2 \\
	 &  \bighs + 4 \nor{e^{\tau \phi} \p \wt{A} \cdot \na \eta}{L^2(B_1)}^2 \\
	 &  = c_{\ep,\delta,\tau}^2 \nor{  e^{\tau \phi}  \Delta \pa{ \eta \p}}{L^2(B_1)}^2 + 4 \nor{e^{\tau \phi} \p \wt{A} \cdot \na \eta}{L^2(B_1)}^2,
 \end{align*}
 where
 \begin{align*}
	 c_{\ep,\delta,\tau}^2 \df 6\kappa_{n} n m^2 \pa{\ep_{n,m,\delta}\pa{\f{4}{\delta}}^{\f{5}{2}} + \f{c}{\tau}}.
 \end{align*}
Eventually, the last inequality yields
 \begin{align*}
	 & \nor{ e^{\tau \phi}  \Delta \pa{\eta \p}}{L^2(B_1)} \nonumber \\
	 &  \bhs \le  \nor{ e^{\tau \phi} \eta  \Delta \p}{L^2(B_1)} + 2 \nor{ e^{\tau \phi} \na \eta \cdot \na \p}{L^2(B_1)} + \nor{ e^{\tau \phi} \p \Delta \eta }{L^2(B_1)}\nonumber \\
	 &  \bhs\le  c_{\ep,\delta,\tau}\nor{  e^{\tau \phi}  \Delta \pa{ \eta \p}}{L^2(B_1)} + 2 \nor{e^{\tau \phi} \p \wt{A} \cdot \na \eta}{L^2(B_1)} \nonumber\\
	 & \bighs\bighs \bhs + 2 \nor{ e^{\tau \phi} \na \eta \cdot \na \p}{L^2(B_1)} + \nor{ e^{\tau \phi} \p \Delta \eta }{L^2(B_1)}.
 \end{align*}
 Now we recall that $\na \eta$ and $\Delta \eta$ are supported in $B_1 \backslash B_{1/2}$ and that they are bounded by a constant independent of $\tau, \ep, \delta$, hence
 \begin{align}\label{est}
	 & \nor{ e^{\tau \phi}  \Delta \pa{\eta \p}}{L^2(B_1)} \nonumber \\
	 &  \bhs\le  c_{\ep,\delta,\tau} \nor{  e^{\tau \phi}  \Delta \pa{ \eta \p}}{L^2(B_1)} + c \nor{e^{\tau \phi} \p \wt{A}  }{L^2\pa{B_1 \backslash B_{1/2}}} \\
	 &  \bighs\bighs \bhs + c \nor{ e^{\tau \phi} \na \p}{L^2\pa{B_1 \backslash B_{1/2}}} + c \nor{ e^{\tau \phi} \p  }{L^2\pa{B_1 \backslash B_{1/2}}}\nonumber \\
	 &  \bhs\le  c_{\ep,\delta,\tau} \nor{  e^{\tau \phi}  \Delta \pa{ \eta \p}}{L^2(B_1)} + c e^{\tau \phi\pa{\ud}}\nor{ \p \wt{A}  }{L^2\pa{B_1 \backslash B_{1/2}}} \\
	 &  \bighs\bighs \bhs + c e^{\tau \phi\pa{\ud}}\nor{  \na \p}{L^2\pa{B_1 \backslash B_{1/2}}} + c e^{\tau \phi\pa{\ud}}\nor{ \p  }{L^2\pa{B_1 \backslash B_{1/2}}}\nonumber \\
	 &  \bhs\le  c_{\ep,\delta,\tau}\nor{  e^{\tau \phi}  \Delta \pa{ \eta \p}}{L^2(B_1)} + c e^{\tau \phi\pa{\ud}},
 \end{align}
 where $c$ does not depend on $\delta, \tau$ or $\ep$, and where we used that $\phi$ is decreasing. We recall that $\delta$ is fixed, and can be taken as small as we want. The constant $\ep_{n,m,\delta}$ needs to be small enough so that $6 \kappa_{n} n m^2 (4/\delta)^{5/2}\ep_{n,m,\delta} <1$. Then $\tau$ needs to be large enough so that $c_{\ep,\delta,\tau} <1$. Then we move the first term of the right hand side of \eqref{est} to the left and get
\begin{align}\label{dern}
\nor{ e^{\tau \phi}  \Delta \pa{\eta \p}}{L^2(B_1)} \le c e^{\tau \phi\pa{\ud}}.
\end{align}
Finally, using our Carleman inequality once more, and because $\phi$ is decreasing, we find
\begin{align*}
	\nor{\p}{L^2(B_{1/2})}& \le \nor{e^{\tau\pa{\phi(\cdot)-\phi\pa{\ud}}} \p}{L^2(B_{1/2})} \\
	& = \nor{e^{\tau\pa{\phi(\cdot)-\phi\pa{\ud}}} \eta \p}{L^2(B_{1/2})} \\
	& \le \nor{e^{\tau\pa{\phi(\cdot)-\phi\pa{\ud}}} \eta\p}{L^2(B_{1})}\\
	& \le c \sqrt{\kappa_n} \tau^{-\f{3}{2}} \nor{e^{\tau\pa{\phi(\cdot)-\phi\pa{\ud}}} \Delta \pa{\eta\p}}{L^2(B_{1})} \\
	& \le c \tau^{-\frac{3}{2}},
\end{align*}
where we used \eqref{dern} is the last step. Letting $\tau \ra + \ii$ proves that $\p = 0$ in $B_{1/2}$. We can propagate this information by a well known argument, see for instance the proof of \cite[Theorem XIII.63]{ReeSim4}. This concludes the proof of Theorem \ref{sucp}. \qed

\section*{Appendix}
The results presented in this appendix are very classical~\cite{Kato,ReeSim2,LieLos01}, and we recall them for completeness. In the proof of \cite[Corollary 1.2]{Garrigue18}, we already recalled how to show that if $v,w \in L^{\f{qd}{2s}}\loc(\R^d)$, then for any $r > 0$ and any $\ep > 0$, there is $c_{\ep,r} \ge 0$ such that 
 \begin{align*}
\ab{v}^q \indic_{B_r} + \ab{w}^q \indic_{B_r} \le \ep (-\Delta)^s + c_{\ep,r} \bhs \tx{in } \R^d,
\end{align*}
in the sense of forms and that this implied that for any $R > 0$ and any $\ep > 0$ there is $c_{\ep,R} \ge 0$ such that
\begin{align*}
\indic_{B_R} \ab{ \sum_{\ell=1}^N v(x_{\ell}) + \sum_{1 \le \ell \sle j \le N} w(x_{\ell}-x_j)}^q \le \ep (-\Delta)^s + c_{\ep,R}\bhs \tx{in } \R^{dN},
\end{align*}
in the sense of forms at the level of many-body operators. Here we present a similar proof in the case of gradient operators, this is the subject of the textbook problem \cite[Problem 37 p343]{ReeSim2}. In some assumptions on magnetic potentials of this document, we could not use the notation $i \indic_{B} A \cdot \na \le \ep (-\Delta) + c$ because $i \indic_{B} A \cdot \na $ is not a symmetric operator. But we show that our assumptions of type \eqref{fifi} are verified when $A \in L^d\loc(\R^d)$.
 \begin{lemma}
 Let $A \in L^d\loc(\R^d,\R^d)$, then for any $\ep > 0$ and any $r \ge 0$, there exists $c_{\ep,r} \ge 0$ such that 
\begin{align}\label{fifi}
	\ab{\ps{u,-i  \indic_{B_r}  A \cdot \na u}}  \le \ps{u,\pa{\ep(-\Delta) + c_{\ep,r}}u} \hs\hs\hs\hs \forall u \in \cC^{\ii}\ind{c}(\R^d).
\end{align}
Moreover, for any $\ep > 0$ and any $R \ge 0$, there exists $c_{\ep,R} \ge 0$ such that
 \begin{align*}
	 \ab{\ps{u,-i\indic_{B_R} \sum_{\ell=1}^N   A(x_{\ell}) \cdot \na_{\ell} u}} \le \ps{u,\pa{\ep(-\Delta) + c_{\ep,R} }u}, \hs\hs\hs\hs\hs \forall u \in \cC^{\ii}\ind{c}(\R^{dN}).
\end{align*}
 \end{lemma}
 \begin{proof}
 Let $d \ge 3$ and $r > 1$. Take $ M > 0$ and let us decompose $A = A \indic_{\ab{A} \ge M} + A \indic_{\ab{A} \sle M}$. For any $u \in \cC^{\ii}\ind{c}(\R^d)$, we have
 \begin{align*}
 & \ab{\ps{u,-i\indic_{B_r} A\cdot \na u}} \le \int \ab{A\indic_{B_{r}}} \ab{ u \na u} \\
 & \bhs \le \int_{ \acs{\ab{A} > M} \cap B_{r} } \ab{A} \ab{u \na u} + M \int  \ab{u \na u} \\
 & \bhs \le \nor{ A \indic_{\ab{A} > M}}{L^d(B_{r})} \nor{ u \na u}{L^{\f{d}{d-1}}} + M^{-1} \int \ab{ \na u}^2 + M^3 \int \ab{u}^2\\
	 & \bhs \le  \nor{ A \indic_{\ab{A} > M}}{L^d(B_{r})} \nor{u}{L^{\f{2d}{d-2}}}\nor{\na u}{L^2} + M^{-1} \int \ab{\na u}^2 + M^{3} \int \ab{ u}^2  \\
	 & \bhs \le \pa{c \nor{ A \indic_{\ab{A} > M}}{L^d(B_{r})}  + M^{-1}} \nor{\na u}{L^2}^2 + M^{3} \int \ab{ u}^2,
 \end{align*}
 where we used H\"older's and Sobolev's inequalities, and where the coefficients are independent of $u$. By dominated convergence, $\nor{A \indic_{\ab{A} > M}}{L^d} \ra 0$ when $M \ra +\ii$, so this proves \eqref{fifi}. For the $N$-body case, we use the previous result on each $N$ components and sum. For $d \in \acs{1,2}$, we have the subcritical Sobolev injections and the argument is the same.
\end{proof}

\bibliographystyle{siam}
\bibliography{biblio}
\end{document}